%% file: main.tex
\pdfoutput=1
\documentclass[11pt,letterpaper]{article}

\usepackage{hchang}
\usepackage{fullpage}
\usepackage{aliascnt}
\usepackage{array,booktabs,float,pdflscape}
\usepackage{algpseudocode}
\usepackage[nameinlink,capitalize,noabbrev]{cleveref}
\hypersetup{hypertexnames=false}

\newcommand{\RR}{\mathbb{R}}

\newcommand{\defn}[1]{\textbf{#1}}
\DeclareMathOperator*{\EX}{\mathbb{E}}

\DeclareMathOperator{\avd}{avd}

\newtheorem{theorem}{Theorem}[section]

\newaliascnt{corollary}{theorem}
\newtheorem{corollary}[corollary]{Corollary}
\aliascntresetthe{corollary}

\newaliascnt{definition}{theorem}
\newtheorem{definition}[definition]{Definition}
\aliascntresetthe{definition}
\newaliascnt{algorithm}{theorem}
\usepackage{algorithm}
\aliascntresetthe{algorithm}
\newaliascnt{lemma}{theorem}
\newtheorem{lemma}[lemma]{Lemma}
\aliascntresetthe{lemma}

\newaliascnt{proposition}{theorem}
\newtheorem{proposition}[proposition]{Proposition}
\aliascntresetthe{proposition}
\crefname{proposition}{Proposition}{Propositions}
\Crefname{proposition}{Proposition}{Propositions}

\crefname{figure}{Algorithm}{Algorithms}
\Crefname{figure}{Algorithm}{Algorithms}

\nolinenumbers

\title{Optimal Average Distortion for Non-Expansive Embeddings of Euclidean Metrics into the Line\\
{\large With applications to cut-matching games and vertex expansion}}
\author{Jack Spalding-Jamieson}
\date{}

\begin{document}

\maketitle

\begin{abstract}
We show that any $n$-point Euclidean metric
admits a non-expansive embedding into the
line with average distortion at most $O\left(\sqrt{\frac{\log n}{\log\log n}}\right)$,
matching a known lower bound.

We also present two applications of this result.
First, we give a cut-player strategy with an optimal spectral guarantee
for the undirected cut-matching game.
Second,
we give an improved spectral approximation for vertex expansion in high-degree graphs.
\end{abstract}

\paragraph{AI Disclosure.} Key technical ideas of this work are due to \EMPH{GPT-6 Astra}, with some simplification due to \EMPH{Claude Opus 5.5}.
The presentation here is our own,
and we take full responsibility for the correctness of our results.

\section{Introduction}

For a metric $(X,d)$,
an \defn{embedding into the line}
is a function $f:X\to\RR$.
We say that such an embedding is \defn{non-expansive}\footnote{This is sometimes called a $1$-Lipschitz embedding.}
if, for every pair $x,y\in X$, $|f(x)-f(y)|\leq d(x,y)$.
The \defn{average distortion} of a non-expansive embedding into the line
is defined as
\[
\avd(f)
 :=\frac{\sum_{x,y\in X}d(x,y)}
         {\sum_{x,y\in X}|f(x)-f(y)|}.
\]
More generally, we define the constant
$\avd((X,d)):=\min_{f:X\to\RR\text{ non-expansive}}\avd(f)$.

Numerous works have studied the average distortion of non-expansive embeddings for a variety of metrics.
In particular, they arise in a number of constructions
for approximation algorithms of graph partitioning problems~\cite{arora2009expander,rabinovich2008average,DRSW27,biswal2010eigenvalue}.
For general $n$-point metrics $(X,d)$, it is known that $\avd((X,d))\in O(\log n)$,
and that this is the best-possible bound for some metrics~\cite{biswal2010eigenvalue}.
For numerous shortest-path metrics $(X,d)$ in restricted graph classes,
it is known that $\avd((X,d))\in O(1)$,
where the constant depends on a class parameter~\cite{rabinovich2008average,biswal2010eigenvalue,conroy2025protect,DRSW27}.
For \emph{squared} Euclidean metrics $(X,\|\cdot\|_2^2)$
(where every triple $x,y,z\in X$ satisfies $\|x-y\|_2^2+\|y-z\|_2^2\geq\|x-z\|_2^2$),
it is known that $\avd((X,\|\cdot\|_2^2))=O(\sqrt{\log n})$~\cite{arora2009expander},
and this bound is tight~\cite{matousek2001dominated}.

For Euclidean metrics $(X,\|\cdot\|_2)$,
it has long been known that $\avd((X,\|\cdot\|_2))\in O(\sqrt{\log n})$~\cite{matousek2001dominated},
and that there are examples where
$\avd((X,\|\cdot\|_2))\in\Omega\left(\sqrt\frac{\log n}{\log\log n}\right)$~\cite{matousek2001dominated}.
Our first main result closes this gap.

\begin{theorem}
\label{thm:embed}
For an integer $k\geq1$, let $X\subset\RR^k$ be a set of size $n$.
Then there exists a non-expansive embedding into the line $f:X\to\RR$
so that $\avd(f)\in O\left(\sqrt{\frac{\log n}{\log\log n}}\right)$.
\end{theorem}

\input{applications}

\subsection{Proof Overview}

The proof of
\cref{thm:embed}
will follow a common two-case structure
employed in various forms in several previous works~\cite{arora2009expander,rabinovich2008average,biswal2010eigenvalue,DRSW27}.
Such proofs first eliminate an easy case.
For convenience,
we use the specific form of this easy case stated
by Davies, Rz{\k{a}}{\.z}ewski, Spalding-Jamieson, and Walczak~\cite{DRSW27}.

\begin{lemma}[{\cite[Lemma 4.1]{DRSW27}}]
\label{lemma:large-core}
Let $(X,d)$ be a metric space,
and let $\Delta:=\frac1{n^2}\sum_{x,y\in X}d(x,y)$.
Let $\kappa>0$ and $c>2$.
If there exists some $x^*\in X$ with
$\left|B\left(x^*,\frac\Delta c\right)\right|\geq\kappa n$,
then the embedding $f(x)=d(x,B(x^*,\Delta/c))$
is non-expansive and
has $\avd(f)\leq\frac1{\frac\kappa 2\left(1-\frac2c\right)}$.
\end{lemma}

In our application of the first lemma,
we will
have $\kappa,c$ be universal constants.
The case when the first lemma does not apply (for constant values $\kappa,c$)
is called the \defn{well-spread} case.
In the well-spread case,
we will instead compute a random projection,
and use it to define a minimum-weight matching problem.
Then, by a standard result from optimal transport,
the dual to this matching problem will give us our
embedding.
The bulk of our work is then to bound
the expected total weight
of the resulting matching.

\section{A Matching Lemma}

The key lemma to our proof is as follows.
\begin{lemma}
\label{lemma:transitions}
Let $X\subset\RR^k$ have size $n$.
For each $g\in\RR^k$, let $K_g$ be a symmetric matrix on $X$
with nonnegative entries and row sums equal to $1$.
Sample $x_0\sim\operatorname{Unif}(X)$ and
$\gamma\sim\mathcal N(0,1)^k$ independently,
then sample $x_1$ from $K_\gamma(x_0,\cdot)$.
Then
\[
\left(\EX|\langle\gamma,x_1-x_0\rangle|\right)^2
=O\left(\EX\|x_1-x_0\|_2^2
\frac{\log n}{\log\log n}\right).
\]
\end{lemma}

\begin{corollary}
\label{cor:separation}
Let $X\subset\RR^k$ be a set of $n$ points,
and let $\Delta:=\frac1{n^2}\sum_{x,y\in X}\|x-y\|_2$.
Assume $|B(x,\Delta/4)|<\frac n{64}$ for every $x\in X$.
For a nonzero vector $g\in\RR^k$,
let $X^-_g,X^+_g\subset X$ be the
$\left\lceil \frac n{64}\right\rceil$
smallest and largest points
in the ordering given by $x\mapsto\langle x,g\rangle$,
respectively.
Let $A_g$ be a minimum-weight bipartite
perfect matching between $X^-_g$ and $X^+_g$,
with weights given by the Euclidean distance between points,
and let $|A_g|$ denote its total weight.
Then,
$\EX_{\gamma\sim\mathcal N(0,1)^k}|A_\gamma|\geq c\cdot n\cdot\Delta\sqrt{\frac{\log\log n}{\log n}}$
for a universal constant $c>0$. %
\end{corollary}

To prove \cref{lemma:transitions}, we will make use of Markov chains.
In particular, we will use
the following result.

\begin{lemma}[{\cite[Theorem~8.12]{ostrovskii2013metric}}]
\label{lemma:euclidean-walk}
Let $x_0,x_1,\ldots$ be a reversible Markov chain on $X\subset\RR^k$,
with $x_0$ sampled from a stationary distribution of the chain.
For every integer $i\geq1$,
\[
\EX\|x_i-x_0\|_2^2
\leq i\EX\|x_1-x_0\|_2^2.
\]
\end{lemma}

At a high level,
the proof of \cref{lemma:transitions} samples a number of different Gaussian variables,
and chains together the resulting matchings to construct a random walk (in the form of a reversible Markov chain).
This on its own is not too dissimilar from some proofs of the $O(\sqrt{\log n})$ bound
for \emph{squared} Euclidean metrics.
In order to obtain our desired bound,
we want to show that the expected total matching weight is large.
We can apply \cref{lemma:euclidean-walk} to
bound the expected squared distance from the starting vertex
at the end of the random walk
using the expected total matching weight.
With the forward walk's starting vertex fixed,
we also generate a second ``reverse'' walk
that attempts to start from a uniformly random vertex,
and tries to arrive at the starting point of the ``forward'' walk.
Unlike the forward walk, the reverse walk uses shifted Gaussians.
An analysis of the success probability of the reverse walk
then eventually leads to our bound.

We defer both proofs to
\cref{sec:lemma-proof}.

\section{Main Embedding Result}

To prove our main embedding result,
we use the following special case
of the Kantorovich-Rubinstein dual.

\begin{lemma}[{\cite[Proposition~2.1 and Section~6.1]{gabriel2019computational}}]
\label{lemma:matching-duality}
Let $(X,d)$ be a metric space of size $n$,
and let $A,B\subset X$ be disjoint sets of equal size.
If $M$ is a minimum-weight perfect matching between $A$ and $B$,
with edge weights $d(x,y)$,
then there exists a non-expansive $f:X\to\RR$ such that
\[
\sum_{y\in B}f(y)-\sum_{x\in A}f(x)
=\sum_{(x,y)\in M}d(x,y).
\]
\end{lemma}

\begin{proofof}{\cref{thm:embed}}
Let $\Delta:=\frac1{n^2}\sum_{x,y\in X}\|x-y\|_2$.
If some ball $B(x,\Delta/4)$ contains at least $n/64$ points,
\cref{lemma:large-core} with $\kappa=1/64$ and $c=4$
gives average distortion at most $256$.
Otherwise, \cref{cor:separation} gives a direction $g$ with
\[
|A_g|\geq cn\Delta\sqrt{\frac{\log\log n}{\log n}}.
\]
By \cref{lemma:matching-duality},
there is a non-expansive $f:X\to\RR$ with
\[
|A_g|=\sum_{y\in X_g^+}f(y)-\sum_{x\in X_g^-}f(x).
\]
Consequently, for each $z\in X$, we have
\[
\begin{aligned}
|A_g|
&=\sum_{y\in X_g^+}(f(y)-f(z))
  +\sum_{x\in X_g^-}(f(z)-f(x))\\
&\leq\sum_{x\in X}|f(x)-f(z)|.
\end{aligned}
\]
Finally,
\[
\avd(f)
=\frac{n^2\Delta}{\sum_{x,z\in X}|f(x)-f(z)|}
\leq\frac{n\Delta}{|A_g|}
\leq\frac1c\sqrt{\frac{\log n}{\log\log n}}. %
\]
\end{proofof}

\input{cutmatch-undirected-ltw}

\input{vertex-expansion}

\bibliographystyle{alpha}
\bibliography{ref}

\appendix

\section{Lemma Proof}
\label{sec:lemma-proof}

\begin{proofof}{\cref{cor:separation}}
We may assume that $n>64$,
and that
$X^-_\gamma$ and $X^+_\gamma$ are always disjoint.

Rescale $X$ so that
$\EX\limits_{\gamma\sim\mathcal N(0,1)^k}|A_\gamma|
=\left\lceil\frac n{64}\right\rceil$.

Sample $\gamma\sim\mathcal N(0,1)^k$,
then sample $uv\sim\operatorname{Unif}(A_\gamma)$.
For each $g\in\RR^k\setminus\{\mathbf 0\}$, define a symmetric transition matrix $K_g$ by
\[
K_g(x,y):=\frac{\mathbf1_{\{xy\in A_g\}}}{1+\|x-y\|_2}
\quad(x\ne y),\qquad
K_g(x,x):=1-\sum_{y\ne x}K_g(x,y).
\]
The diagonal entries are nonnegative because each vertex has at most one partner.

Independently sample $x,y\sim\operatorname{Unif}(X)$.
$X^-_\gamma$ and $X^+_\gamma$ each contain $\lceil n/64\rceil$ points.
If neither $x$ nor $y$ belongs to $X^-_\gamma\cup X^+_\gamma$,
then
$|\langle\gamma,x-y\rangle|\leq|\langle\gamma,u-v\rangle|$.
With these distributions we obtain
\[
\begin{alignedat}{2}
\Pr\left(|\langle\gamma,u-v\rangle|\geq\frac\Delta8\right)
&\geq\Pr
\left(|\langle\gamma,x-y\rangle|\geq\frac\Delta8\right)
-2\cdot\frac{2\lceil n/64\rceil}{n}
&&\quad\text{by the union bound for $x,y$}\\
&\geq\frac12\Pr
\left(\|x-y\|_2\geq\frac\Delta4\right)-\frac18
&&\quad\text{since the $\mathcal N(0,1)$ density is $<1/2$}\\
&\geq\frac12\left(1-\frac1{64}\right)-\frac18\geq\frac14
&&\quad\text{by the well-spread assumption}.
\end{alignedat}
\]
Let
\begin{equation}\label{eq:matching-progress}
\begin{alignedat}{2}
t&:=\frac2n\EX\limits_{\gamma\sim\mathcal N(0,1)^k}
\sum_{uv\in A_\gamma}
\frac{|\langle\gamma,u-v\rangle|}{1+\|u-v\|_2}\\
&\geq\frac{2\lceil n/64\rceil}{n}
\frac{\left(\EX\limits_{\substack{\gamma\sim\mathcal N(0,1)^k\\ %
uv\sim\operatorname{Unif}(A_\gamma)}}
\sqrt{|\langle\gamma,u-v\rangle|}\right)^2}
{\EX\limits_{\substack{\gamma\sim\mathcal N(0,1)^k\\
uv\sim\operatorname{Unif}(A_\gamma)}}(1+\|u-v\|_2)}
&&\quad\text{by Cauchy-Schwarz}\\
&=\frac{2\lceil n/64\rceil}{n}
\frac{\left(\EX\limits_{\substack{\gamma\sim\mathcal N(0,1)^k\\
uv\sim\operatorname{Unif}(A_\gamma)}}
\sqrt{|\langle\gamma,u-v\rangle|}\right)^2}{2}\\
&\geq\frac{\lceil n/64\rceil}{n}\frac\Delta8
\Pr\left(|\langle\gamma,u-v\rangle|\geq\Delta/8\right)^2\\
&\geq\frac1{64}\frac\Delta8(1/4)^2
=\frac\Delta{8192}.
\end{alignedat}
\end{equation}
Also,
\begin{equation}\label{eq:step-length}
\frac2n\EX\limits_{\gamma\sim\mathcal N(0,1)^k}
\sum_{uv\in A_\gamma}\frac{\|u-v\|_2^2}{1+\|u-v\|_2}
\leq\frac2n\EX\limits_{\gamma\sim\mathcal N(0,1)^k}|A_\gamma|
=
\frac{2\lceil n/64\rceil}{n}\leq1.
\end{equation}

Independently sample $x_0\sim\operatorname{Unif}(X)$,
then sample $x_1$ from $K_\gamma(x_0,\cdot)$.
For this sampling,
\begin{alignat}{2}
\EX|\langle\gamma,x_1-x_0\rangle|
&=\frac1n\EX\limits_{\gamma\sim\mathcal N(0,1)^k}
\sum_{x,y\in X}K_\gamma(x,y)|\langle\gamma,y-x\rangle|\notag\\
&=\frac2n\EX\limits_{\gamma\sim\mathcal N(0,1)^k}
\sum_{uv\in A_\gamma}
\frac{|\langle\gamma,u-v\rangle|}{1+\|u-v\|_2}=t
&&
\label{eq:sampled-progress}\\
\EX\|x_1-x_0\|_2^2
&=\frac1n\EX\limits_{\gamma\sim\mathcal N(0,1)^k}
\sum_{x,y\in X}K_\gamma(x,y)\|y-x\|_2^2\notag\\
&=\frac2n\EX\limits_{\gamma\sim\mathcal N(0,1)^k}
\sum_{uv\in A_\gamma}\frac{\|u-v\|_2^2}{1+\|u-v\|_2}
&&
\notag\\
&\leq1
&&\quad\text{by \eqref{eq:step-length}}.\label{eq:sampled-step-length}
\end{alignat}
Thus,
\[
\begin{alignedat}{2}
t^2
&=\left(\EX|\langle\gamma,x_1-x_0\rangle|\right)^2
&&\quad\text{by \eqref{eq:sampled-progress}}\\
&=O\left(\EX\|x_1-x_0\|_2^2
\frac{\log n}{\log\log n}\right)
&&\quad\text{by \cref{lemma:transitions} applied to $K_g$}\\
&=O\left(\frac{\log n}{\log\log n}\right)
&&\quad\text{by \eqref{eq:sampled-step-length}}.
\end{alignedat}
\]
Finally,
\[
\begin{alignedat}{2}
\frac{\EX_{\gamma\sim\mathcal N(0,1)^k}|A_\gamma|}{n\Delta}
&=
\frac{\lceil n/64\rceil}{n\Delta}
\geq\frac1{64\Delta}\\
&\geq\frac1{64\cdot8192t}
&&\quad\text{by \eqref{eq:matching-progress}}\\
&=\Omega\left(\sqrt{\frac{\log\log n}{\log n}}\right).
\end{alignedat}
\]
\end{proofof}

\begin{proofof}{\cref{lemma:transitions}}
Rescale $X$ so that $\EX\|x_1-x_0\|_2^2=1$,
and let $t:=\EX|\langle\gamma,x_1-x_0\rangle|$.
Denote $r:=\lceil\log n\rceil$.
Sample $x_0$ uniformly from $X$ and independently sample vectors
$\gamma_1,\ldots,\gamma_r\sim\mathcal N(0,1)^k$.
For $i=1,\ldots,r$, sample $x_i$ from $K_{\gamma_i}(x_{i-1},\cdot)$.
Let $\pi$ denote the resulting joint distribution of
$(x_0,\gamma_1,x_1,\ldots,\gamma_r,x_r)$.
Since every $K_g$ is symmetric, the chain is reversible
and the uniform distribution is stationary.
For $1\leq i\leq r$, stationarity makes $x_{i-1}$ uniform on $X$.
Since $\gamma_i$ is sampled independently of $x_{i-1}$,
\begin{equation}\label{eq:walk-progress}
\EX\limits_\pi|\langle\gamma_i,x_i-x_{i-1}\rangle|
=t.
\end{equation}
Also,
\begin{equation}\label{eq:walk-distance}
\EX\limits_\pi\|x_i-x_0\|_2^2
\leq i\EX\limits_\pi\|x_1-x_0\|_2^2\leq i
\qquad\text{by \cref{lemma:euclidean-walk}}.
\end{equation}

Fix $a\in X$.
Let $\rho_a$ be the joint distribution of
$(y_0,\eta_1,y_1,\ldots,\eta_r,y_r)$
generated by the following procedure:
\begin{enumerate}
\item Sample $y_r$ uniformly from $X$.
\item For $i=r,r-1,\ldots,1$, independently choose a plus or minus sign
with equal probability.
Draw an independent vector from $\mathcal N(0,1)^k$ and add
$\pm t(y_i-a)/i$ to obtain $\eta_i$, %
using the chosen sign.
Then sample
\[
y_{i-1}\sim K_{\eta_i}(y_i,\cdot).
\]
\end{enumerate}
The probability that the second procedure ends at $a$ is at most $1$.
With $z_0=a$, its expression in terms of the original walk is
\[
\begin{aligned}
&\Pr\limits_{\rho_a}(y_0=a)\\
&\quad=\frac1n\sum_{z_1,\ldots,z_r\in X}\int_{(\RR^k)^r}
\prod_{i=1}^r\left[
K_{g_i}(z_i,z_{i-1})
\frac{e^{-\frac12\|g_i-\frac ti(z_i-a)\|_2^2}
+e^{-\frac12\|g_i+\frac ti(z_i-a)\|_2^2}}
{2(2\pi)^{k/2}}\right]dg_1\cdots dg_r\\
&\quad=\frac1n\sum_{z_1,\ldots,z_r\in X}\int_{(\RR^k)^r}
\prod_{i=1}^r\left[
K_{g_i}(z_{i-1},z_i)
\frac{e^{-\|g_i\|_2^2/2}}{(2\pi)^{k/2}}
\frac{e^{-\frac12\|g_i-\frac ti(z_i-a)\|_2^2}
+e^{-\frac12\|g_i+\frac ti(z_i-a)\|_2^2}}
{2e^{-\|g_i\|_2^2/2}}\right]dg_1\cdots dg_r\\
&\quad=\EX\limits_\pi\left[
\frac1n\prod_{i=1}^r
\frac{e^{-\frac12\|\gamma_i-\frac ti(x_i-a)\|_2^2}
+e^{-\frac12\|\gamma_i+\frac ti(x_i-a)\|_2^2}}
{2e^{-\|\gamma_i\|_2^2/2}}
\,\middle|\,x_0=a\right].
\end{aligned}
\]
Let $f_a$ denote the positive random variable inside this expectation.
Then $\EX_\pi[f_a\mid x_0=a]\leq1$.
Moreover,
\begin{equation}\label{eq:gaussian-ratio}
\begin{aligned}
\log f_a
&=-\log n+\sum_{i=1}^r\log\left[
\frac{e^{-\|\gamma_i\|_2^2/2}
\left(e^{\frac ti\langle\gamma_i,x_i-a\rangle}
+e^{-\frac ti\langle\gamma_i,x_i-a\rangle}\right)
e^{-\frac{t^2}{2i^2}\|x_i-a\|_2^2}}
{2e^{-\|\gamma_i\|_2^2/2}}\right]\\
&=-\log n+\sum_{i=1}^r\left[
\log\frac{e^{\frac ti\langle \gamma_i,x_i-a\rangle}
+e^{-\frac ti\langle \gamma_i,x_i-a\rangle}}2
-\frac{t^2}{2i^2}\|x_i-a\|_2^2\right]\\
&\geq-\log n-r\log2+\sum_{i=1}^r\left[
\frac ti|\langle \gamma_i,x_i-a\rangle|
-\frac{t^2}{2i^2}\|x_i-a\|_2^2\right].
\end{aligned}
\end{equation}
Consequently,
\[
\EX\limits_\pi[\log f_a\mid x_0=a]
\leq\log\EX\limits_\pi[f_a\mid x_0=a]\leq0
\qquad\text{by Jensen's inequality}.
\]
Averaging over $a\in X$,
\[
\begin{alignedat}{2}
\log n+r\log2
&\geq\sum_{i=1}^r\left[\frac ti\EX\limits_\pi
|\langle\gamma_i,x_i-x_0\rangle|
-\frac{t^2}{2i^2}\EX\limits_\pi\|x_i-x_0\|_2^2\right]
&&\quad\text{by \eqref{eq:gaussian-ratio}}\\
&\geq\sum_{i=1}^r\left[
\frac ti\left(\EX\limits_\pi|\langle\gamma_i,x_i-x_{i-1}\rangle|
-\EX\limits_\pi|\langle\gamma_i,x_{i-1}-x_0\rangle|\right)
-\frac{t^2}{2i}\right]
&&\quad\begin{aligned}&\text{by the triangle inequality}\\
&\text{and \eqref{eq:walk-distance}}\end{aligned}\\
&\geq\sum_{i=1}^r\left[\frac{t^2}{2i}
-\frac ti\EX\limits_\pi
|\langle\gamma_i,x_{i-1}-x_0\rangle|\right]
&&\quad\text{by \eqref{eq:walk-progress}}\\
&\geq\sum_{i=1}^r\left[\frac{t^2}{2i}
-\frac ti\sqrt{\EX\limits_\pi\|x_{i-1}-x_0\|_2^2}\right]
&&\quad\begin{aligned}&\text{by Cauchy-Schwarz}\\
&\text{and independence of $\gamma_i$}\\
&\text{from $(x_0,x_{i-1})$}\end{aligned}\\
&\geq\sum_{i=1}^r\left[\frac{t^2}{2i}-\frac t{\sqrt i}\right]
&&\quad\text{by \eqref{eq:walk-distance}}\\
&=\sum_{i=1}^r\left[\frac{t^2}{4i}-1
+\left(\frac t{2\sqrt i}-1\right)^2\right]\\
&\geq\sum_{i=1}^r\left[\frac{t^2}{4i}-1\right]
&&\quad\text{since $\left(\frac t{2\sqrt i}-1\right)^2\geq0$}\\
&=\frac{t^2}{4}\sum_{i=1}^r\frac1i-r
&&\quad\text{since $\sum_{i=1}^r1=r$}.
\end{alignedat}
\]
Consequently,
\[
\frac{t^2}{4}\log\log n
\leq\frac{t^2}{4}\int_1^{r+1}\frac{du}{u}
\leq\frac{t^2}{4}\sum_{i=1}^r\frac1i
\leq\log n+r(1+\log2)
\leq(3+2\log2)\log n.
\]
\aftermath
\end{proofof}

\end{document}

%% file: applications.tex
\subsection{Applications}

We present two applications of our main embedding result.

\input{cutmatch-intro-undirected}

\subsubsection{Vertex Expansion}

Let $G=(V,E)$ be a graph with $n$ vertices.
For a subset $A\subset V$ with $|A|\leq n/2$,
let $\partial A$ denote the set of vertices in $V\setminus A$ adjacent to some vertex in $A$.
The \defn{vertex expansion} of a vertex set $A$ in $G$ is the ratio $\psi(A):=\frac{|\partial A|}{|A|}$.
The \defn{vertex expansion of a graph} $G$ is the minimum vertex expansion over all possible sets $|A|\leq\frac n2$,
and it is denoted as $\psi(G)$.

A recent line of work concerns a ``Cheeger-style'' approximation
of vertex expansion.
In particular, the following result was obtained by two groups independently in 2022~\cite{kwok2022cheeger,kwok2025cheeger,jain2022dimension}.
\begin{theorem}[{\cite{kwok2022cheeger,kwok2025cheeger,jain2022dimension}}]
\label{thm:vtx-cheeger}
For a graph $G=(V,E)$ with maximum degree $\Delta$,
there is a polynomial time algorithm
that produces a set $S\subset V$
with
$$\psi(S)\in O\left(\sqrt{\psi(G)\cdot\log\Delta}\right).$$
\end{theorem}

We improve this result for graphs with large maximum degree.
\begin{theorem}
\label{thm:dense-vtx-cheeger}
For a graph $G=(V,E)$ with $n$ vertices,
there is a randomized polynomial time algorithm
that, with probability at least $2/3$, produces a set $S\subset V$
with
$$\psi(S)\in O\left(\sqrt{\psi(G)\cdot\frac{\log n}{\log\log n}}\right).$$
\end{theorem}

%% file: cutmatch-intro-undirected.tex
\subsubsection{Cut-Matching Games}

In 2006, Khandekar, Rao, and Vazirani introduced the \defn{cut-matching game}~\cite{khandekar2009graph},
which has been used in various forms
to efficiently approximate several graph partitioning problems~\cite{orecchia2008partitioning,louis2010cut,lau2024fast}.

\begin{definition}\label{def:undirected-cut-matching}
Fix a vertex set $V$ of (even) size $n$.
Two players alternate for $T$ rounds: A \defn{cut player} and a \defn{matching player}.
In round $t$ of the \defn{cut-matching game}:
\begin{enumerate}
\item The cut player chooses a bisection $(S_t,V\setminus S_t)$.
\item The matching player chooses a perfect matching $M_t$ across this bisection.
\end{enumerate}
After $T$ rounds, let $H:=(V,M_1\cup M_2\cup\cdots\cup M_T)$
be the graph consisting of the union of all these matchings,
retaining repeated edges.
\end{definition}

In the original formulation,
the cut player seeks to maximize the \defn{edge expansion}
of $H$,
given by
\[
\phi(H):=\min_{\emptyset\ne S\subsetneq V}
\frac{|E_H(S,V\setminus S)|}{\min\{|S|,|V\setminus S|\}},
\]
regardless of the matching player's choices.

Khandekar, Rao, and Vazirani originally
proved that a cut player could achieve
$\phi(H)=\Omega(1)$ after $O(\log^2n)$ iterations (with high probability)~\cite{khandekar2009graph}.

In 2008,
Orecchia, Schulman, Vazirani, and Vishnoi
improved this to
$\phi(H)=\Omega(\log n)$ after $O(\log^2n)$ iterations (with constant probability)~\cite{orecchia2008partitioning}.
In fact, they showed that the cut player can achieve the same guarantee
on a spectral quantity:
Let $L(H)$ denote the Laplacian matrix of $H$,
and let $\lambda_2(H)$ denote its second-smallest eigenvalue.
Note that $\phi(H)\geq\frac12\lambda_2(H)$.
Their result actually guarantees the stronger bound of $\lambda_2(H)=\Omega(\log n)$.

Later, Sherman showed that this improved bound was nearly optimal~\cite{sherman2009breaking}.
In particular, he showed that
the matching player can force
\[
\lambda_2(H)=O\left(T\frac{\log\log n}{\log n}\right).
\]
We tighten the cut player strategy to make this bound tight.
In fact, doing so using \cref{thm:embed}
would be quite straightforward.
However, the resulting algorithm would be quite slow.
Instead, with more work, we are able to prove the existence of
a cut player strategy that achieves this bound
using near-linear work per round.

\begin{theorem}\label{thm:undirected-cut-matching}
For the cut-matching game,
there is a (randomized) cut-player strategy using near-linear work per round
that, against every matching player, forces
\[
\phi(H)\geq\frac12\lambda_2(H)=\Omega(\log n),
\]
in $O(\log^2 n/\log\log n)$ rounds,
with probability at least $2/3$.
\end{theorem}

%% file: cutmatch-undirected-ltw.tex
\section{Cut-Matching Games}

We essentially provide a new dimension-reduction
algorithm that could be applied to modifications of several
different algorithms.
We specifically choose to modify the algorithm
of Lau, Tung, and Wang~\cite{lau2024fast}
as given in simplified form by Lau in his recent textbook~\cite[Section~18.3]{lau2027algorithmic}.
Like Lau, we choose to initially present a slow form of the algorithm
that computes exact matrix exponentials.

\begin{algorithm}[H]
\caption{Cut player}
\label{alg:undirected-cut-player}
\algrenewcommand\algorithmicrequire{\textbf{Input:}}
\algrenewcommand\algorithmicensure{\textbf{Output:}}
\begin{algorithmic}[1]
\Require An ordered set $V$ of even size $n$, a positive integer $T$,
and access to a matching player.
\Ensure The union $H_T$ of the matching replies, retaining repeated edges.
\State Initialize $H_0$ as the empty graph on $V$.
\For{$t=1,\ldots,T$}
\State Set $W_t\gets e^{-L(H_{t-1})/4}$ and $P_t\gets W_t/(\operatorname{Tr}(W_t)-1)$.
\State Let $z_v\in\RR^V$ be row $v$ of
$P_t^{1/2}(I-\mathbf1\mathbf1^T/n)$, where $\mathbf1\in\RR^V$ is the all-ones vector.
\State Flip a fair coin.
\State \qquad If heads, set $y_v\gets\|z_v\|_2$ for all $v$.
\State \qquad If tails, sample an independent $\gamma\sim\mathcal N(0,I_n)$,
and set $y_v\gets\langle\gamma,z_v\rangle$ for all $v$.
\State Query the bisection given by the $n/2$ smallest values of $y_v$,
breaking ties by the order on $V$.
Receive $M_t$ and set $H_t\gets H_{t-1}+M_t$.
\EndFor
\State \Return $H_T$.
\end{algorithmic}
\end{algorithm}

The algorithm as presented by Lau is equivalent to always taking the ``tails'' branch.
Hence, much of the analysis will still immediately apply,
particularly regarding the fast implementation of this algorithm.
We first require two lemmas.

\begin{lemma}
\label{lemma:scalar-median}
Let $(y_v)_{v\in V}$ be real numbers with $\sum_{v\in V}y_v=0$,
where $|V|=n$ is even.
For every perfect matching $M$ between the smallest and largest halves,
\[
\sum_{uv\in M}(y_u-y_v)^2\geq\sum_{v\in V}y_v^2.
\]
\end{lemma}
\begin{proof}
Following the proof of \cite[Lemma~17.10]{lau2027algorithmic},
choose $m$ between the two halves and write each edge $uv\in M$
with $y_u\leq m\leq y_v$.
Since $m-y_u$ and $y_v-m$ are nonnegative,
\[
(y_u-y_v)^2
=(y_u-m)^2+(y_v-m)^2+2(m-y_u)(y_v-m)
\geq(y_u-m)^2+(y_v-m)^2.
\]
Summing over $M$ gives
\[
\begin{aligned}
\sum_{uv\in M}(y_u-y_v)^2
&\geq\sum_{v\in V}(y_v-m)^2\\
&=\sum_{v\in V}y_v^2-2m\sum_{v\in V}y_v+nm^2\\
&=\sum_{v\in V}y_v^2+nm^2
&&\text{since $\sum_{v\in V}y_v=0$}\\
&\geq\sum_{v\in V}y_v^2.
\end{aligned}
\]
\end{proof}

Our change to the dimension reduction gives the following guarantee,
independently of how the vectors are computed.

\begin{lemma}\label{lemma:cut-dimension-reduction}
Let $(z_v)_{v\in V}$ be vectors in $\RR^k$ with $\sum_{v\in V}z_v=\mathbf 0$,
where $|V|=n$ is even.
Flip a fair coin.
If heads, let $S$ consist of the $n/2$ vertices with smallest $\|z_v\|_2$.
If tails, sample $\gamma\sim\mathcal N(0,I_k)$ and let $S$ consist of the
$n/2$ vertices with smallest $\langle\gamma,z_v\rangle$.
For each outcome, choose any perfect matching $M$ between $S$ and $V\setminus S$.
Then
\[
\EX\left[\sum_{uv\in M}\|z_u-z_v\|_2^2\right]
\geq c\frac{\log\log n}{\log n}\sum_{v\in V}\|z_v\|_2^2
\]
for a universal $c>0$.
\end{lemma}
\begin{proof}
Let $M_r$ be the matching chosen when the coin is heads.
For each $g\in\RR^k$, let $M_g$ be the matching chosen when the coin is tails
and $\gamma=g$.
Let $r_v:=\|z_v\|_2$ and $\bar r:=\frac1n\sum_{v\in V}r_v$.
As we will see, the norm bisection controls the variance of $r_v$,
while the Gaussian bisection controls its mean.
We have
\begin{equation}\label{eq:cut-norm}
\begin{aligned}
\sum_{v\in V}(r_v-\bar r)^2
&\leq\sum_{uv\in M_r}(r_u-r_v)^2
&&\text{by \cref{lemma:scalar-median} for $r_v-\bar r$}\\
&\leq\sum_{uv\in M_r}\|z_u-z_v\|_2^2
&&\text{by the triangle inequality}.
\end{aligned}
\end{equation}

For a fixed $g\in\RR^k$, let $m_g$ be a median of
$(\langle g,z_v\rangle)_{v\in V}$.
Then
\begin{equation}\label{eq:cut-median-shift}
\begin{aligned}
n\cdot|m_g|
&=\left|\sum_{v\in V}(m_g-\langle g,z_v\rangle)\right|
&&\text{since $\sum_{v\in V}z_v=\mathbf 0$}\\
&\leq\sum_{v\in V}|\langle g,z_v\rangle-m_g|
&&\text{by the triangle inequality}.
\end{aligned}
\end{equation}
Thus
\begin{equation}\label{eq:cut-median}
\begin{aligned}
\sum_{v\in V}|\langle g,z_v\rangle|
&\leq\sum_{v\in V}|\langle g,z_v\rangle-m_g|+n\cdot|m_g|
&&\text{by the triangle inequality}\\
&\leq2\sum_{v\in V}|\langle g,z_v\rangle-m_g|
&&\text{by \eqref{eq:cut-median-shift}}\\
&=2\sum_{uv\in M_g}|\langle g,z_u-z_v\rangle|.
\end{aligned}
\end{equation}
Taking expectations, we obtain
\begin{equation}\label{eq:cut-projection}
\begin{aligned}
\sqrt{\frac2\pi}\sum_{v\in V}r_v
&=\EX_{\gamma\sim\mathcal N(0,I_k)}\left[\sum_{v\in V}|\langle\gamma,z_v\rangle|\right]
&&\text{since $\langle\gamma,z_v\rangle\sim\mathcal N(0,r_v^2)$}\\
&\leq2\EX_{\gamma\sim\mathcal N(0,I_k)}
\left[\sum_{uv\in M_\gamma}|\langle\gamma,z_u-z_v\rangle|\right]
&&\text{by \eqref{eq:cut-median}}.
\end{aligned}
\end{equation}

Let $K_g(u,v)=1$ for $uv\in M_g$ and $0$ otherwise.
Independently of $\gamma$, choose a uniform vertex $u\in V$,
let $v$ be its partner in $M_\gamma$, and set $x_0=z_u$, $x_1=z_v$.
Then
\begin{equation}\label{eq:cut-matching-application}
\begin{aligned}
\left(\EX\left[|\langle\gamma,x_1-x_0\rangle|\right]\right)^2
&=O\left(\frac{\log n}{\log\log n}\cdot
\EX\left[\|x_1-x_0\|_2^2\right]\right)
&&\text{by \cref{lemma:transitions}}.
\end{aligned}
\end{equation}
Consequently,
\begin{equation}\label{eq:cut-mean}
\begin{aligned}
\frac1n\left(\sum_{v\in V}r_v\right)^2
&\leq\frac{\pi n}{2}\cdot
\left(\EX\left[|\langle\gamma,x_1-x_0\rangle|\right]\right)^2
&&\text{by \eqref{eq:cut-projection}}\\
&=O\left(\frac{\pi n}{2}\cdot\frac{\log n}{\log\log n}\cdot
\EX\left[\|x_1-x_0\|_2^2\right]\right)
&&\text{by \eqref{eq:cut-matching-application}}\\
&=O\left(\frac{\pi n}{2}\cdot\frac{\log n}{\log\log n}\cdot\frac2n\cdot
\EX_{\gamma\sim\mathcal N(0,I_k)}
\left[\sum_{uv\in M_\gamma}\|z_u-z_v\|_2^2\right]\right)\\
&=O\left(\frac{\log n}{\log\log n}\cdot
\EX_{\gamma\sim\mathcal N(0,I_k)}
\left[\sum_{uv\in M_\gamma}\|z_u-z_v\|_2^2\right]\right).
\end{aligned}
\end{equation}
Since the coin is fair,
\begin{equation}\label{eq:cut-coin}
\EX\left[\sum_{uv\in M}\|z_u-z_v\|_2^2\right]
=\frac12\sum_{uv\in M_r}\|z_u-z_v\|_2^2
+\frac12\EX_{\gamma\sim\mathcal N(0,I_k)}
\left[\sum_{uv\in M_\gamma}\|z_u-z_v\|_2^2\right].
\end{equation}
Consequently,
\[
\begin{aligned}
\sum_{v\in V}\|z_v\|_2^2
&=\sum_{v\in V}r_v^2\\
&=\sum_{v\in V}(r_v-\bar r)^2+\frac1n\left(\sum_{v\in V}r_v\right)^2\\
&\leq\sum_{uv\in M_r}\|z_u-z_v\|_2^2
+O\left(\frac{\log n}{\log\log n}\right)
\EX_{\gamma\sim\mathcal N(0,I_k)}
\left[\sum_{uv\in M_\gamma}\|z_u-z_v\|_2^2\right]
&&\text{by \eqref{eq:cut-norm} and \eqref{eq:cut-mean}}\\
&=O\left(\frac{\log n}{\log\log n}
\EX\left[\sum_{uv\in M}\|z_u-z_v\|_2^2\right]\right)
&&\text{by \eqref{eq:cut-coin}.\qed}
\end{aligned}
\]
\end{proof}

\begin{proofof}{\cref{thm:undirected-cut-matching}}
Take $T=\lceil C\log^2n/\log\log n\rceil$ for a sufficiently large universal $C$.
As in Lau's analysis, in each round the exact vectors satisfy
$\sum_{v\in V}z_v=\mathbf 0$ and
$\sum_{v\in V}\|z_v\|_2^2=\operatorname{Tr}(P_t(I-\mathbf1\mathbf1^T/n))=1$.
Hence
\[
0\leq\operatorname{Tr}(P_tL(M_t))
=\sum_{uv\in M_t}\|z_u-z_v\|_2^2
\leq2\sum_{v\in V}\|z_v\|_2^2=2.
\]
The MMWU regret bound~\cite[Theorem~18.5 and equation~(18.2)]{lau2027algorithmic} gives
\begin{equation}\label{eq:cut-regret}
\lambda_2(H_T)\geq\frac12\sum_{t=1}^T\operatorname{Tr}(P_tL(M_t))-4\log n.
\end{equation}
By \cref{lemma:cut-dimension-reduction}, for a universal $c>0$,
conditional on all preceding rounds,
\begin{equation}\label{eq:cut-round-progress}
\EX[\operatorname{Tr}(P_tL(M_t))]\geq c\frac{\log\log n}{\log n}.
\end{equation}
Let $\tau$ be the first round with
$\sum_{t=1}^{\tau}\operatorname{Tr}(P_tL(M_t))\geq10\log n$,
or $T$ if there is none.
The event $\tau\geq t$ is determined by the preceding rounds, so
\[
\begin{aligned}
c\frac{\log\log n}{\log n}\EX[\tau]
&\leq\EX\left[\sum_{t=1}^{\tau}\operatorname{Tr}(P_tL(M_t))\right]
&&\text{by \eqref{eq:cut-round-progress}}\\
&\leq10\log n+2.
\end{aligned}
\]
Thus $\EX[\tau]\leq T/3$ for sufficiently large $C$.
Markov's inequality gives
$\Pr[\tau\geq T]\leq\EX[\tau]/T\leq1/3$.
Consequently, with probability at least $2/3$,
\begin{equation}\label{eq:cut-total-progress}
\sum_{t=1}^T\operatorname{Tr}(P_tL(M_t))\geq10\log n.
\end{equation}
On this event,
\[
\begin{aligned}
2\phi(H_T)&\geq\lambda_2(H_T)\\
&\geq\frac12\sum_{t=1}^T\operatorname{Tr}(P_tL(M_t))-4\log n
&&\text{by \eqref{eq:cut-regret}}\\
&\geq5\log n-4\log n=\log n
&&\text{by \eqref{eq:cut-total-progress}}.
\end{aligned}
\]

The standard fast MMWU technique
via Johnson-Lindenstrauss
extends this same argument
to a near-linear time algorithm
with slightly worse constant,
with constant probability
(see e.g.~\cite[Section~18.3 and Theorem~9.14]{lau2027algorithmic}).
\end{proofof}

%% file: vertex-expansion.tex
\section{Vertex Expansion}

In this section, we prove
\cref{thm:dense-vtx-cheeger}.
Doing so will involve the following (known) parametrized optimization problem.
\begin{equation}
\label{eq:gamma-d}
\everymath{\displaystyle}
\begin{array}{r c c r c l l}
\gamma^{(d)}(G) & := & \min_{\substack{f:V\to\RR^d\\y:V\to\RR_{\geq0}}}
& \dfrac{\sum_{v\in V}y(v)}{\sum_{x\in V}\|f(x)\|_2^2} \\
&& \textup{subject to} & \sum_{v\in V} f(v) & = & \mathbf 0\\
&& & y(u)+y(v) & \geq & \|f(u)-f(v)\|_2^2 & \forall uv\in E.
\end{array}
\end{equation}
Several important facts about $\gamma$ are already well-known, for different values of $d$.
First and foremost, $\gamma^{(n)}$ is a semidefinite program,
and can be solved (approximately) in polynomial time
(see~\cite{spalding2025reweighted} for some discussion of specific time complexities).
Second is the following result relating $\gamma^{(1)}$ to vertex expansion.

\begin{theorem}[{\cite[Proposition 2.2 and Theorem 2.10]{olesker2022geometric,olesker2024geometric}}]
\label{thm:cheeger-vertex-one-dim}
For a graph $G$,
$$\psi(G)^2\leq O\left(\gamma^{(1)}(G)\right)\leq O\left(\psi(G)\right).$$
\end{theorem}
The sweep construction in the cited proof finds, from any feasible
one-dimensional pair $(f,y)$, a set $S$ in polynomial time with
$\psi(S)^2=O\left(\sum_{v\in V}y(v)/\sum_{v\in V}f(v)^2\right)$.
Note that $\gamma^{(n)}(G)\leq\gamma^{(1)}(G)$
by definition (it is a relaxation).
Hence,
to prove \cref{thm:dense-vtx-cheeger},
it suffices to prove the following result.

\begin{proposition}
\label{thm:degree-dim-red}
For a graph $G$ with $n$ vertices,
$$\gamma^{(1)}(G)\in O\left(\gamma^{(n)}(G)\cdot\frac{\log n}{\log\log n}\right).$$
A randomized polynomial time algorithm realizes this bound
with probability at least $2/3$.
\end{proposition}

To do this, we will require a slightly different kind of metric embedding result
than what
\cref{thm:embed} gives.
Specifically,
\cref{thm:embed}
bounds the average distortion,
while we require a bound on a quantity we call the \defn{average squared distortion},
which (for a non-expansive metric embedding $f$ into the line of a metric $(X,d)$)
is defined as
\[
\avd^2(f)
 :=\frac{\sum_{x,y\in X}d(x,y)^2}
         {\sum_{x,y\in X}|f(x)-f(y)|^2}.
\]
This quantity has seen use in prior works for different purposes~\cite{biswal2008eigenvalue,biswal2010eigenvalue,spalding2025reweighted}.
It turns out that our existing bound on the average distortion already gives a nice bound on the average squared distortion, through a known relationship.

\begin{lemma}[{\cite[special case from proof of equation~(7.33)]{Naor2014comparison}}]\label{lem:average-powers}
\label{lemma:sqr-avd}
Let $(X,d)$ be an $n$-point metric space, and let $h:X\to\RR$
be a non-expansive embedding of average distortion $D$.
Then there exists a non-expansive $f:X\to\RR$ satisfying
\[
\sum_{x,y\in X}d(x,y)^2
\leq CD^2\sum_{x,y\in X}|f(x)-f(y)|^2
\]
for a universal constant $C$.
Given $h$ and the pairwise distances, $f$ can be computed in polynomial time.
\end{lemma}

\begin{proofof}{\cref{thm:degree-dim-red}}
Let $f,y$ be an optimal pair for $\gamma^{(n)}(G)$.
Embed the points given by $f$ into the line using
\cref{thm:embed}
and
\cref{lemma:sqr-avd}
to get a non-expansive embedding $f':V\to\RR$
into the line
with average squared distortion at most $O\left(\frac{\log n}{\log\log n}\right)$.
Assume without loss of generality (by shifting) that
$\sum_{v\in V}f'(v)=0$.
Consequently, by non-expansitivity,
$f',y$ must be a feasible solution to $\gamma^{(1)}(G)$,
and it must have objective value
at most
\[
\begin{aligned}
\gamma^{(1)}(G)
&\leq\frac{\sum_{v\in V}y(v)}{\sum_{v\in V}f'(v)^2}\\
&=\frac{2n\sum_{v\in V}y(v)}{\sum_{u,v\in V}|f'(u)-f'(v)|^2}\\
&\leq O\left(\frac{\log n}{\log\log n}\cdot
\frac{2n\sum_{v\in V}y(v)}{\sum_{u,v\in V}\|f(u)-f(v)\|_2^2}\right)\\
&\leq O\left(\frac{\log n}{\log\log n}\cdot
\frac{\sum_{v\in V}y(v)}{\sum_{v\in V}\|f(v)\|_2^2}\right)\\
&=O\left(\frac{\log n}{\log\log n}\gamma^{(n)}(G)\right).
\end{aligned}
\]
\aftermath
\end{proofof}